\documentclass[smallcondensed]{svjour3}
\usepackage[utf8]{inputenc}
\usepackage{mathptmx}  
\usepackage{latexsym}
\usepackage{natbib} 
 \usepackage{amsmath} 
 \usepackage{amsfonts}
 \usepackage{amssymb} 
\usepackage{graphicx}

 \usepackage{hyperref}   

\newtheorem{Theorem}{Theorem}[section]
\newtheorem{Proposition}[Theorem]{Proposition}

\newcommand{\CH}{\mathbb{C}_h^{n\times n}} 
\newcommand{\CM}{\mathbb{C}^{n\times n}}

\newcommand{\reals}{\mathbb{R}}

\newcommand{\gambles}{\mathcal{H}}

\newcommand{\domain}{\mathcal{K}}

\renewcommand{\epsilon}{\varepsilon}

\begin{document}


\title{A Gleason-type theorem for any dimension based on a gambling formulation of Quantum Mechanics}

\author{Alessio Benavoli\\ Alessandro Facchini\\ Marco Zaffalon}

\institute{Istituto Dalle Molle di Studi sull'Intelligenza Artificiale (IDSIA) Lugano (Switzerland).
\email{$\{$alessio,alessandro.facchini,zaffalon$\}$@idsia.ch}}

\date{Received: date / Accepted: date}

\maketitle

\begin{abstract}
Based on a gambling formulation of quantum mechanics, 
we derive a Gleason-type theorem that holds for any dimension $n$ of a quantum system, and in particular for $n=2$. The theorem states that the only logically consistent probability assignments are exactly the ones that are definable as the trace of the product of a projector and a density matrix operator. In addition, we detail the reason why dispersion-free probabilities are actually not valid, or rational, probabilities for quantum mechanics, and hence should be excluded from consideration.
\end{abstract}



%
%

\section{Introduction}
Born's rule is a law of  Quantum Mechanics (QM) that gives the probability that a measurement on a quantum system will yield a given result.
This rule was initially considered  a postulate of QM, but then it was investigated whether it could actually be derived as a theorem from other postulates of QM. Given reasonable assumptions, \citet{gleason1957measures} showed that probabilities in QM are expressible as the trace of the product of a projector and
a density operator.  Gleason's theorem is nowadays one of the fundamental theorems in QM.

The issue with Gleason's proof is that it is long and complicated; and it only  applies to quantum systems of dimension greater than two. 
For $n=2$, Gleason's theorem leaves room for the existence of probabilities that are not expressible as the trace of the product of a projector and
a density operator. An example of these probabilities are the so-called \textit{dispersion-free probabilities} \cite{bell1966problem} \cite[Sec. 6]{kochen1967problem}; these probabilities only take
values $0$ or $1$. 
An open question is  whether they are \textit{valid} probabilities, or more generally whether there is any valid probability that is not expressible as the trace of the product of a projector and a density operator. 

\subsection{Digression about classical probability}
\label{sec:introprob}
In classical subjective, or Bayesian, probability, there is a well-established way to check whether  the probability assignments of a certain subject, whom we call Alice, about the result of an uncertain experiment is valid, in the sense that they are self-consistent---we call them \textit{coherent}. The idea is to use these probability assignments to define odds, which are the inverses of probabilities, about the  results of the experiment (e.g., Head or Tail in the case of a coin toss) and show that there is no gamble in this betting system that leads Alice to a sure loss, that is, to losing money no matter the outcome of the experiment. Historically this is also referred to as the impossibility to make a Dutch book; and Alice is regarded as coherent, or rational, if she chooses her odds so that no Dutch book is possible.
\Citet{finetti1937} showed that Kolmogorov's probability axioms can be derived by imposing the principle of coherence alone on a subject's odds about an uncertain experiment. 

\citet{williams1975} and \citet{walley1991} have later shown that it is possible to follow de Finetti's tradition and ideas to justify probability while making things even simpler and more elegant. Their approach is also more general than de Finetti's, because coherence is defined without any explicit reference to probability (which is also what allows coherence to be generalised to other domains, such as quantum mechanics); the idea is to work in a dual space of gambles. To understand this framework, we consider an experiment whose outcome  $\omega$ belongs to a certain  space of possibilities $\Omega$ (e.g., Head or Tail). We can model Alice's beliefs about $\omega$ by asking her whether she accepts engaging
in certain risky transactions, called \textit{gambles}, whose outcome depends on the actual
outcome of the experiment. Mathematically, a gamble is a bounded real-valued function on $\Omega$, $g:\Omega 
\rightarrow \mathbb{R}$, which is interpreted as an uncertain reward in a linear utility scale. If Alice accepts a gamble $g$, this means that she commits herself to 
receive $g(\omega)$ \emph{utiles}\footnote{Abstract units of utility, indicating the satisfaction derived from an economic transaction.} if the outcome of the experiment eventually happens 
to be the event $\omega \in \Omega$. Since $g(\omega)$ can be negative, Alice can also lose utiles. Therefore Alice's acceptability of a gamble depends on her knowledge about the experiment.

The  set of gambles that Alice accepts is called her set of \emph{desirable} (or \emph{acceptable}) \emph{gambles}. 
One such set is said to be \emph{coherent} when it satisfies the following criteria:\footnote{For an example see \citet[Example 1]{benavoli2016d}.}
\begin{enumerate}
 \item Any gamble $g\neq0 $ such that $g(\omega)\geq0$ for each $\omega \in \Omega$ must be desirable for Alice, given 
that it may increase Alice's capital without ever decreasing it
 (\textbf{accepting partial gain)}. 
 \item Any gamble $g$ such that $g(\omega)\leq0$ for each $\omega \in \Omega$ must not be desirable for Alice, given 
that it may only decrease Alice's capital without ever increasing it  (\textbf{avoiding partial loss}). 
 \item If Alice finds $g$ to be desirable, 
 then also $\lambda g$ must be desirable for her for any $0<\lambda \in \mathbb{R}$ (\textbf{positive homogeneity}).
\item If Alice finds $g_1$ and $g_2$ desirable, 
then she also must accept $g_1+g_2$ 
(\textbf{additivity}). 
\end{enumerate}

\noindent Note how these four axioms express some truly minimal requirements: the first means that Alice likes to increase her wealth; the second that she does not like to decrease it; the third and fourth together simply rephrase the assumption that Alice's utility scale is linear. 

In spite of the simple character of these requirements, these four axioms alone define a very general theory of probability. De Finetti's (Bayesian) theory is the particular case obtained by additionally imposing some regularity (continuity) requirement and especially completeness, that is, the idea that a subject should always be capable of comparing options \cite{williams1975,walley1991}.\footnote{By enforcing those requirements, partial and sure (Dutch book) losses become equivalent.} On the other hand, \citet{zaffalon2015a} have shown that Axioms~1--4 above are equivalent, under the assumption of linearity of utility, to the decision-theoretic axiomatisation of incomplete preferences in the classical tradition of \citet{anscombe1963} (see, e.g., \citet{galaabaatar2013}). 

This is to say that Axioms~1--4 have a very long history as a solid axiomatic foundation of rationality. And it is precisely their emphasis on rationality that allows us to connect them to quantum mechanics, as we describe next.

\subsection{Back to quantum mechanics}
In \citet{benavoli2016d}, we have generalised Williams-Walley's gambling system to QM and shown that, by imposing the same rationality criteria, 
it is possible to derive all the four postulates of QM as a consequence of rational gambling on a quantum experiment.
This is tantamount to showing that QM is the Bayesian theory generalised to the space of Hermitian matrices. 
It has ``probabilities'' (density matrices), Bayes' rule (measurement), marginalisation (partial tracing), independence (tensor product). 
In this framework, QM probability assignments computed via Born's rule are derived quantities, and they are proved to represent \textit{
valid (coherent/self-consistent/rational)}  probability assignments.

The present paper uses these results to: 
\begin{enumerate}
 \item show that  dispersion-free probabilities are \textit{incoherent} and
  \item derive a stronger version of Gleason's theorem that holds in any dimension, through much a simpler proof, which states that: all \textit{coherent} probability assignments in QM  must be obtained as the trace of the product of a projector and a density operator.
\end{enumerate}

\subsection{Related work}
In recent times, several attempts have been made to find a Gleason-type theorem whose applicability covers the two-dimensional case too. 
A notable generalisation of Gleason's theorem holding in any finitely dimensional Hilbert space  has been obtained by relaxing the orthogonality constraint on the measurement operators. More specifically, observables are identified with positive operator valued measurements (POVMs)  \cite{holevo1982probabilistic}, which are defined by any partition of the identity operator into a set of $\ell$ positive operators $E_i$---called \emph{effects}---, acting on an $n$-dimensional Hilbert space and representing the outcomes of the measurement.  A generalized probability measure is then defined on the set of all effects, that is, the positive operators that can occur in the range of a POVM. All such generalized probability measures are then proven to be of the standard form, i.e., determined by a density operator ($p(E_i) = \text{trace}(\rho E_i)$). The one-to-one relationship between generalized probability measures on the effects and density operators  was first derived
by  \citet{holevo1973statistical} (see also \cite[Sec. 1.6.1.]{holevo1982probabilistic}). More recently, \citet{busch2003quantum}  has re-derived it with the goal of providing a simplified proof of Gleason’s theorem that also covers the two-dimensional case. 
The same approach has been  pursued by \citet{caves2004gleason}, and a further generalisation is provided by   \citet{barnett2014quantum}. 
 
However, probability measures defined on effects are peculiar. POVMs can in fact be regarded as
imperfect observations, since they are not repeatable. 


On another side, the idea of justifying  QM from rationality principles on a gambling system was proposed in the Bayesian interpretation  of QM (QBism), see for instance \citet{Caves02,longPaper, FuchsSchackII,mermin2014physics} and  Pitowsky's quantum gambles \cite{pitowsky2003betting,Pitowsky2006}. These attempts have mostly focused on the probabilities that can be derived from QM; we argue instead that QM itself is the Bayesian theory of probability---extended so as to enable it to deal with quantum experiments. For a detailed discussion about similarities and differences between our framework and QBism we refer to \citet[Section 8]{benavoli2016d}.

%
%
%
%
%
%
%
%
%
%
%

\section{Quantum gambling system}
We start by defining a gambling system about the results of a quantum experiment. 
To this end, we consider two subjects: the bookmaker and the gambler (Alice).  The bookmaker prepares the quantum system in some quantum state. Alice
has her personal knowledge (beliefs) about the experiment---possibly no knowledge at all.

\begin{enumerate}
\item   {
The bookmaker 
 announces that he will measure the quantum system along  its $n$ orthogonal directions and so the outcome of the measurement is an element of $\Omega=\{\omega_1,\dots,\omega_n\}$,  with $\omega_i$ denoting the elementary event ``detection along $i$''. 
Mathematically,  it means that the quantum system is measured along its eigenvectors,\footnote{We mean the eigenvectors of the density matrix of the quantum 
system.} i.e., the projectors\footnote{A projector $\Pi$ is a set of $n$ positive semi-definite matrices in $\CH$ such that   $\Pi_i\Pi_k=0$, $(\Pi_i)^2=\Pi_i=(\Pi_i)^\dagger$,  $\sum_{i=1}^n \Pi_i=I$.} $\Pi^*=\{\Pi^*_{1},\dots,\Pi^*_{n}\}$
and $\omega_i$ is the event ``indicated'' by the $i$-th projector. The bookmaker is fair, meaning that he will correctly perform the experiment and report
 the actual results to Alice.} 
  \item {Before the experiment, Alice declares the set of gambles she is willing to accept.  Mathematically, a gamble $G$ on this experiment 
is a Hermitian matrix in $\CM$, the space of all Hermitian  $n \times n$ matrices being denoted by $\CH$.  We will denote the set of gambles Alice is willing to accept by $\domain \subseteq \CH$.}
\item By accepting  a gamble $G$, Alice commits herself to receive  $\gamma_{i}\in \reals$ utiles  if the outcome of the experiment eventually happens to be 
$\omega_i$. The value $\gamma_{i}$ is defined from $G$ and $\Pi^{*}$ as follows:
 \begin{equation}
  \Pi^{*}_{i}G\Pi^{*}_{i}=\gamma_{i}\Pi^{*}_{i} \text{ for } i=1,\dots,n.
 \end{equation} 
It is a real number since $G$ is Hermitian.
\end{enumerate}
Denote by $\gambles^+=\{G\in\CH:G\gneq0\}$ the
subset of all positive semi-definite and non-zero (PSDNZ) matrices  in $\CH$: we call them the set of \emph{positive gambles}.
The set of negative gambles is similarly given by $\gambles^-=\{G\in\CH:G\lneq0\}$.  Alice examines the gambles in $\CH$ and comes up with the subset $\domain$ of the gambles that she finds desirable. Alice's rationality is characterised as follows.
\begin{definition}[Rationality criteria]\label{def:axQM}
\begin{enumerate}
 \item Any gamble $G \in \CH$ such that $G \gneq0$ must be desirable for Alice, given that it may increase Alice's 
utiles without ever decreasing them
 (\textbf{accepting partial gain}). This means that $ \gambles^+ \subseteq \domain$.
\item Any gamble $G \in \CH$ such that $G \lneq0$ must not be desirable for Alice, given that it may only decrease 
Alice's utiles without ever increasing them  (\textbf{avoiding partial loss}). This means that $ \gambles^- \cap 
\domain=\emptyset$.
\item If Alice finds $G$ desirable, that is
$G \in \domain$, then also $\nu G$ must be desirable for her for any $0<\nu \in \reals$ (\textbf{positive homogeneity}).
\item If Alice finds $G_1$ and $G_2$ desirable, that is
$G_1,G_2 \in \domain$, then she must also accept $G_1+G_2$, i.e., $G_1+G_2 \in \domain$ (\textbf{additivity}). 
\end{enumerate}
\end{definition}
To understand these rationality criteria, originally presented in \citet[Sec. III]{benavoli2016d}, we must remember that mathematically the payoff for any gamble $G$
is computed as $\Pi_i^{*} G \Pi_i^{*}$ if the outcome of the experiment is the event indicated by $\Pi_i^{*}$.
Then the first two rationality criteria above hold no matter the experiment $\Pi^{*}$ that 
is eventually performed. 
In fact,  from the properties of PSDNZ matrices, if 
 $G \gneq0$ then  $\Pi_i^{*} G \Pi_i^{*}=\gamma_{i} \Pi_i^{*}$ with $\gamma^{*}_{i}\geq0$ for any $i$ and 
$\gamma_{j}>0$ for some $j$. Therefore, by accepting $G \gneq0$, Alice can only increase her utiles.
Symmetrically,  if $G \lneq0$ then $\Pi_i^{*} G \Pi_i^{*} = \gamma_{i} 
\Pi_i^{*}$ with $\gamma_{i}\leq 0$ for any $i$. 
Alice must then avoid the gambles $G \lneq0$ because they can only decrease her utiles.
This justifies  the first two rationality criteria. 
 For the last two, observe that 
 $$
 \Pi_i^{*} (G_1+G_2) \Pi_i^{*}=\Pi_i^{*} G_1 \Pi_i^{*}+\Pi_i^{*} G_2 \Pi_i^{*}=(\gamma_i+\vartheta_i) \Pi_i^{*},
 $$ 
 where we have  exploited the fact that $\Pi_i^{*} G_1 \Pi_i^{*}=\gamma_i  \Pi_i^{*}$ and $\Pi_i^{*} G_2 
\Pi_i^{*}=\vartheta_i \Pi_i^{*}$. Hence, if Alice accepts $G_1,G_2$, she must also accept $G_1+G_2$ because this 
will lead to a reward of $\gamma_i+\vartheta_i$.
 Similarly, if $G$ is desirable for Alice, then also $\Pi_i^{*} (\nu G) \Pi_i^{*}=
 \nu\Pi_i^{*}  G \Pi_i^{*}$ should be desirable for any $\nu>0$. 
 
In other words, as in the case of classical desirability described in Section \ref{sec:introprob}, the four conditions above state only minimal requirements: that Alice would like to increase her wealth and not decrease it (conditions $1$ and $2$); and that Alice's utility scale is linear (conditions $3$ and $4$). The first two conditions should be plainly uncontroversial. The linearity of the utility scale is routinely assumed in the theories of personal, and in particular Bayesian, probability as a way to isolate  considerations of uncertainty from those of value (wealth).

%
%

We can characterise $\domain$ also from a geometric point of view. In fact, from the above properties, it follows that Alice's set of desirable gambles mathematically satisfies the following properties.
 
\begin{definition}
  \label{def:sdg}
Let $\domain $ be a subset of $\CH$. We say that  $\domain$ is   a {\bf coherent set of strictly desirable gambles (SDG)} if
\begin{description}
 \item[(S1)] $\domain$ is a non-pointed convex-cone (positive homogeneity and additivity);
 \item[(S2)] if $G\gneq0$ then $G \in \domain$ (accepting partial gain);
   \item[(S3)] if $G \in \domain$ then either $G \gneq0$ or $G -\epsilon I \in  \domain$ for some strictly positive real number $\epsilon$ (openness).\footnote{In \citet{benavoli2016d} we used another formulation of openness, namely (S3'): if $G \in \domain$ then either $G \gneq0$ or $G -\Delta \in  \domain$ for some $0<\Delta \in \CH$. (S3) and (S3') are provably equivalent given (S1) and (S2).}
\end{description}
\end{definition}
\noindent Note that the although the additional openness property (S3) is not necessary for rationality, it is technically convenient as it precisely isolates the kind of models we use in QM (as well as in classical probability) \cite{benavoli2016d}.
Property (S3) has a  gambling interpretation too: it means that we will only consider gambles that are \emph{strictly} desirable for Alice; these are the gambles for which Alice expects gaining something---even an epsilon of utiles.
Note that assumptions (S1) and (S2)  imply that SDG also avoids partial loss: if $G \lneq 0$, then $G \notin \domain$
\cite[Remark III.2]{benavoli2016d}. 


\begin{definition}
An SDG is said to be maximal if there is no larger SDG containing it. 
\end{definition}

In \citet{benavoli2016d}, we have shown that maximal SDGs and density matrices are one-to-one.
The mapping between them is obtained through the standard inner product in $\CH$, i.e.,
$G\cdot R= Tr(G^\dagger R)$ with $G,R\in \CH$.
This follows by a  representation result whose proof is a direct application of Hahn-Banach theorem:

\begin{Theorem}[Representation theorem from \citet{benavoli2016d}]\label{thm:repr}
\label{thm:dualityopen} 
The map that associates a maximal SDG the unique density matrix $\rho$ such that 
\begin{equation}
Tr(G^\dagger \rho) \geq  0~ \forall G \in \domain\end{equation} 
defines a bijective correspondence between maximal SDGs and  density matrices. 
 Its inverse is the map $(\cdot)^\circ$ that associates each density matrix $\rho$ the  maximal SDG\footnote{Here the gambles  $G\gneq0 $ are treated separately because they are always desirable and, thus, they are not informative on Alice's beliefs
about the quantum system. Alice's knowledge is determined by the gambles that are not $G\gneq0 $.}
\begin{equation}
\label{eq:induced}
(\rho)^\circ=\{ G \in \CH \mid G  \gneq0\}\cup\{G \in \CH \mid  Tr(G^\dagger \rho) > 0\}. 
\end{equation}

\end{Theorem}
This representation theorem has several consequences. First, it provides a gambling interpretation of the first axiom of QM on density operators. Second, it shows that density operators are coherent, since
the set $(\rho)^\circ$ that they induce in \eqref{eq:induced} is a valid SDG. This also implies that QM is self-consistent---a gambler that uses QM to place bets on a quantum experiment cannot be made a partial (and, thus, sure) loser.  
Third, the first axiom of QM on $\mathbb{C}_h^{n \times n}$ is structurally and formally equivalent to 
Kolmogorov's first and second axioms about probabilities on $\mathbb{R}^n$ \cite[Sec. 2]{benavoli2016d}. In fact, they can be both derived via duality
from a coherent set of desirable gambles on $\mathbb{C}_h^{n \times n}$ and, respectively, $\mathbb{R}^n$.
In \citet{benavoli2016d} we have also derived Born's rule and the other three axioms of QM as a consequence of rational gambling on a quantum experiment
and show that that measurement, partial tracing and tensor product are equivalent to the probabilistic notions of
Bayes' rule, marginalisation and independence.
We will now use these results to prove the two main results of the present paper.

%
%

\section{The incoherence of dispersion-free  probability measures}

Let  $\mathfrak{P}(\mathbb{C}_h^{n \times n})$ be the lattice of orthogonal projectors in $\mathbb{C}_h^{n \times n}$.
Gleason's theorem relies on the fact that a function $p: \mathfrak{P}(\mathbb{C}_h^{n \times n}) \to [0,1]$ to be called a probability measure has to satisfy the following properties:
\begin{description}
\item[(P1)] $p(I_n)=1$,  
\item[(P2)] $p(\sum^m_{i=1} \Pi_i)=\sum^m_{i=1} p(\Pi_i)$, for each sequence $(\Pi_1, \dots, \Pi_m)$ of mutually orthogonal projectors, and $m\leq n$.
\end{description}
Property (P2), usually called \emph{non-contextuality}, asserts that the probability measure associated to a Hilbert subspace is independent of the choice of the basis.  It implies that, given a Hermitian matrix $G \in \mathbb{C}_h^{n \times n}$, if  
$G=\sum_{i=1}^m \lambda_i \Pi_i$ and $G=\sum_{\ell=1}^k \gamma_\ell \Pi'_\ell$ are two different decompositions obtained from the Spectral Decomposition Theorem,\footnote{This happens when an eigevalue has multiplicity greater than one.} then 
$\sum_{i=1}^m \lambda_i p(\Pi_i) = \sum_{\ell=1}^k \gamma_\ell p(\Pi'_\ell)$.

The crucial point is whether Properties (P1) and (P2) are strong enough to characterise only valid probabilities---that those two conditions appear as a straightforward extension of the classical probability axioms to QM does not mean that they represent the most appropriate way to define probabilities in such a generalised setting.



Our standpoint, as it follows from de Finetti's investigation, is that the essence of probability is the idea of rationality (self-consistency). This is captured by the coherence axioms, which are more primitive than the probability axioms; in this sense, they are better candidates to extend the probability to QM.

To verify this idea, we first need to define what is a coherent probability. We start by defining the expectation associated to a probability measure $p$ as
\[
E_p(G)=\sum_{i=1}^n \lambda_i p(\Pi_i),
\]
where $ \sum_{i=1}^n \lambda_i \Pi_i$ is the decomposition of $G$.

The set of desirable gambles associated to $p$ is thus  defined as
\[
\mathcal{K}_p= \{G \in \mathbb{C}_h^{n \times n} \mid G \gneq0 \textit{ or }E_p(G)>0\},
\]
i.e.,  all $G \gneq0 $ (that are always desirable) and all $G$ whose expectation 
w.r.t.\ $p$ is greater than zero.
Therefore, we say that:
\begin{definition}
A probability measure $p: \mathfrak{P}(\mathbb{C}_h^{n \times n}) \to [0,1]$ is {\bf coherent} if $\mathcal{K}_p$ is a coherent set of strictly desirable gambles.
\end{definition}
\noindent Thus, from the Bayesian perspective adopted here, a probability measure $p$ satisfying (P1) and (P2) is coherent whenever the set of desirable gambles $\mathcal{K}_p$ implied by $p$ is a coherent SDG.

Consider the Hilbert space $\mathbb{C}_h^{2 \times 2}$, and in particular dispersion-free  probability measures, which are those that assign only the values 0 and 1. 

To define them, we exploit the fact that any projector can be written as
$$
\Pi_n=\frac{1}{2}(I+n\cdot \sigma),
$$
where $n\in \mathbb{R}^3$, $||n||=1$ and $\sigma$ is the Pauli matrices basis.
Its orthogonal vector is $\Pi_{-n}$, since
$\Pi_n+\Pi_{-n}=I$, and $\Pi_n\Pi_{-n}=0$.
Note that any $0\neq G \in \mathbb{C}_h^{2 \times 2}$ can be uniquely decomposed, i.e.,  there is a unique projector  $\Pi_n$ such that $G=\lambda_1 \Pi_n+\lambda_2 \Pi_{-n}$, for some $\lambda_1, \lambda_2 \in \reals$. 

 \begin{definition} A 2D dispersion-free probability measure is any function $p: \mathfrak{P}(\mathbb{C}_h^{2 \times 2}) \to [0,1]$
 defined as 
\begin{equation}
\label{eq:dispersionp}
p(\Pi_n)=\left\{\begin{array}{ll}
                 1 & \textit{if } n\cdot x>0 \textit{ or}\\
                  & \textit{if } n\cdot x=0, ~n\cdot y>0\textit{ or}\\
                   & \textit{if } n\cdot x=0, ~n\cdot y=0, ~n\cdot z>0,\\
                   0 & \textit{ otherwise,}
                \end{array}
\right.
\end{equation} 
for some orthogonal vectors $x,y,z\in\mathbb{R}^3$.
\end{definition}

We now show with an example that dispersion-free probabilities are incoherent. 

\begin{example}
Without loss of generality let us assume that $x=(1,0,0)$, $y=(0,1,0)$, $z=(0,0,1)$. 
Consider then the gamble 
$$
G=\lambda_1 \Pi_g+\lambda_2 \Pi_{-g}=\begin{bmatrix}
\begin{array}{cc}
 -\frac{1}{4} & \frac{5}{4} \left(1-i \sqrt{2}\right) \\
 \frac{5}{4} \left(1+i \sqrt{2}\right) & -\frac{11}{4} \\
\end{array}
\end{bmatrix},
$$
with $\lambda_1=1$, $\lambda_2=-4$ and $g=(1/2,1/\sqrt{2},1/2)$.
This gamble is clearly desirable, and thus in $\domain_p$, since $g\cdot x>0$ and $\lambda_1=1>0$ and so
 $E_p(G)=1$.

Consider now the gamble 
$$
H=\gamma_1 \Pi_h+\gamma_2 \Pi_{-h}=\begin{bmatrix}
\begin{array}{cc}
 -\frac{11}{4} & \frac{5}{4} \left(1+i \sqrt{2}\right) \\
 \frac{5}{4} \left(1-i \sqrt{2}\right) & -\frac{1}{4} \\
\end{array}
\end{bmatrix},
$$
with
 $\gamma_1=1$, $\gamma_2=-4$ and $h=(1/2,-1/\sqrt{2},-1/2)$.
 This gamble is also  desirable since $h\cdot x>0$ and $\gamma_1=1>0$ and so
 $E_p(H)=1$, meaning $H \in \domain_p$.
 
Let $F=G+H$. 
 Notice that
$$
F=\displaystyle{\begin{bmatrix}
   -3 & \frac{5}{2}\vspace{1mm}\\ \frac{5}{2} & -3\\
   \end{bmatrix}},
$$
which we can rewrite as $\varrho_1 \Pi_f+\varrho_2 \Pi_{-f}$
with $\varrho_1=-\frac{1}{2}$, $\varrho_2=-\frac{11}{2}$
and $f=(1, 0, 0)$ and so $F<0$. 
Since $f\cdot x>0$, we have that $E[F]=-\frac{1}{2}<0$. 
As a result Alice is incoherent. In fact by accepting $G$ and $H$, she should also be willing to accept $F$ by additivity (Axiom~3 in Definition~\ref{def:axQM}). But this means that Alice would incur a sure loss (a Dutch book), which is a strong form of irrationality. $\lozenge$
\end{example}
Stated differently, the example shows that dispersion-free probabilities are logically inconsistent with the axioms of QM and therefore should not be regarded as, nor called, probabilities. In the next section we detail the formal argument.



\section{A subjective extension of Gleason's theorem}
We will now show that the only coherent probabilities are those that can be defined as the trace of the product of a projector and a density operator (it is well known that dispersion-free probabilities cannot be defined in this way; see, e.g., \citet{heinosaari2011mathematical}).

We first characterise sets of gambles defined by coherent probability measures.
\begin{Proposition}\label{prop:coherent}
If $p: \mathfrak{P}(\mathbb{C}_h^{n \times n}) \to [0,1]$ is a coherent probability measure, then $\mathcal{K}_p$ is a maximal SDG.
\end{Proposition}
\begin{proof}
Since $p$ is coherent, we know that $\domain_p$ is a coherent set of desirable gambles. It is then enough to check that it satisfies the openness condition (S3) and maximality. For openness, assume $G=\sum_i^n \lambda_i \Pi_i \in \domain_p$, meaning $E_p(G) > 0$. Consider any $\epsilon > 0$ such that $E_p(G) -\epsilon > 0$
and let $F=G- \epsilon I=\sum_i^n (\lambda_i -\epsilon) \Pi_i$.
Then $E_p(F)=\sum_i^n (\lambda_i -\epsilon) p(\Pi_i)=E_p(G) - \epsilon(\sum_i^n p(\Pi_i))= E_p(G) - \epsilon p(I_2)=E_p(G) - \epsilon > 0$. Hence $F \in \domain_p$. 

For maximality, we reason as follows. Assume once more $G=\sum_i^n \lambda_i \Pi_i \notin \domain_p$, meaning $E_p(G) \leq 0$. Assume moreover that there is an SDG $\domain' \supset \domain_p$ such that $G \in \domain' $. By openness, $F=G -\epsilon I \in  \domain'$ for some $\epsilon > 0$. This means that $E_p(F)<0$, and thus $- F\in K_p$. 
Since $\domain'$ is an SDG, then $-F+F =0 \in \domain'$, a contradiction. 
\end{proof}

Everything is now ready to obtain the following subjective extension of Gleason's theorem: 

\begin{Theorem}\label{cor:gleason}
For every $n>0$, a probability measure $p: \mathfrak{P}(\mathbb{C}_h^{n \times n}) \to [0,1]$ is coherent if and only if it is of the form \[ p(\Pi_\ell)= \text{Tr}(\Pi_\ell \rho ), \]
for a unique density matrix $\rho \in \mathbb{C}_h^{n \times n} $.  
\end{Theorem}
\begin{proof}
The fact that $p(\Pi_\ell)$ defined from $\rho$ via $\text{Tr}(\Pi_\ell \rho )$ is a coherent probability measure follows from Theorem \ref{thm:repr} and Proposition \ref{prop:coherent}. 

Conversely, let $p$ be a coherent probability measure.
From Theorem \ref{thm:repr} and Proposition \ref{prop:coherent}, there is a unique $\rho$ such that $\domain_p= (\rho)^\circ$. Assume there is a projector $\Pi_\ell$ such that $p(\Pi_\ell) \neq \text{Tr}(\Pi_\ell \rho )$. Observe that the case $\text{Tr}(\Pi_\ell \rho )<p(\Pi_\ell)$ can be led back to that with the reversed inequality by considering the orthogonal projector $(I_n - \Pi_\ell)$.

Let us therefore consider the case that $\text{Tr}(\Pi_\ell \rho )>p(\Pi_\ell)\geq 0$. Assume that $p(\Pi_\ell)=E_p(\Pi_\ell)=0$, whence $p(I_n - \Pi_\ell)=1$. Fix $\epsilon>0$ and let $G= \lambda_1 \Pi_\ell + \lambda_2(I_n - \Pi_\ell)$, where $$\lambda_1=\epsilon + \frac{1}{\text{Tr}(\Pi_\ell \rho )}\text{ and }\lambda_2= \begin{cases} -1 & \text{ if } \text{Tr}((I_n -\Pi_\ell) \rho )=0,\\
-\frac{1}{\text{Tr}((I_n-\Pi_\ell) \rho )} & \text{ otherwise}.\end{cases}$$
It holds that $G\not\geq 0$ and $E_p(G)<0$. 
This means that $G\notin \domain_p$. On the other hand, $$Tr(G^\dagger \rho)={\lambda_1}{\text{Tr}(\Pi_\ell \rho )} +{\lambda_2 }{\text{Tr}((I_n - \Pi_\ell) \rho )}=\epsilon\text{Tr}(\Pi_\ell \rho )+
\begin{cases} 
1 & 
\text{ if } \text{Tr}((I_n -\Pi_\ell) \rho )=0,\\
0 & \text{ otherwise}.\end{cases}$$
Hence $Tr(G^\dagger \rho)> 0$, meaning that $G \in (\rho)^\circ$, a contradiction.

Finally, let us consider the case that $\text{Tr}(\Pi_\ell \rho )>p(\Pi_\ell)>0$. Notice that $\text{Tr}((I_n - \Pi_\ell) \rho ) < p(I_n - \Pi_\ell)$. In this case too $\text{Tr}((I_n - \Pi_\ell) \rho ) > 0$, otherwise, by reasoning as before, we could find a contradiction.  
Hence, we obtain that $\frac{\text{Tr}(\Pi_\ell \rho )}{p(\Pi_\ell)} > 1$ and $\frac{\text{Tr}((I_n - \Pi_\ell) \rho )}{p(I_n - \Pi_\ell) } < 1$. Let $\lambda = \frac{1}{p(\Pi_\ell)}$,  $\mu=-\frac{1}{p(I_n - \Pi_\ell)}$ and $G= \lambda \Pi_\ell + \mu (I_n - \Pi_\ell)$. Given that $G\not\geq 0$ and that, trivially, $E_p(G)=0$, it holds that $G \notin K_p$. But $Tr(G^\dagger \rho)=\frac{\text{Tr}(\Pi_\ell \rho )}{p(\Pi_\ell)} - \frac{\text{Tr}((I_n - \Pi_\ell) \rho )}{p(I_n - \Pi_\ell) } > 0$, meaning that $G \in (\rho)^\circ$, a contradiction. 
\end{proof}


%
%

\section{Conclusion}



We have shown that a subject who uses dispersion-free probabilities to accept gambles on a quantum experiment can always be made a \emph{sure loser}: she loses utiles no matter the outcome of the experiment. We say that dispersion-free probabilities are incoherent, which means that they are logically inconsistent with the axioms of QM. 
Moreover, using such a betting framework, we have proved that the only logically consistent probabilities are those that agree with Born rule. We find remarkable that these results are obtained by only using logical considerations and very simple arguments. This has been made possible by our recent reformulation of QM as a logic of uncertainty for quantum experiments.

The results obtained in this paper extend the scope of Gleason's theorem \citep{gleason1957measures} to any dimension. We believe that our formulation of the theorem does not suffer from problematic interpretations like those, for instance, of the extension  presented by \citet{busch2003quantum} that employs positive operator valued measurements (POVMs). Indeed, we work directly  with the space of projectors and use  a well-established way from classical approaches to rationality to check whether some given probabilistic assessments are self-consistent.  As future work, we plan to  use this gambling interpretation of QM to investigate the validity of the various  hidden variable models. In particular, we aim at exploring how far one can go in defining hidden variable models while staying rational.

\bibliographystyle{spbasic}
\bibliography{biblio}

\end{document}